\documentclass[10pt]{amsart}
\usepackage{amssymb}
\usepackage{amsmath,amsthm}
\usepackage{amscd}
\usepackage{enumitem}
\usepackage{cancel}
\usepackage{pdfsync}
\usepackage{mathrsfs}
\usepackage{stmaryrd}
\usepackage[colorlinks,linkcolor=blue,citecolor=blue,urlcolor=blue]{hyperref}

\newcommand{\beq}{\begin{equation}}
\newcommand{\eeq}{\end{equation}}
\newcommand{\bea}{\begin{eqnarray}}
\newcommand{\eea}{\end{eqnarray}}
\newcommand{\nn}{\nonumber}

\newcommand{\bk}{\begin{cases}}
\newcommand{\ek}{\end{cases}}

\newtheorem{definition}{Definition}
\newtheorem{proposition}{Proposition}
\newtheorem{theorem}{Theorem}
\newtheorem{corollary}{Corollary}

\theoremstyle{definition}

\newtheorem{remark}{\textbf{Remark}}

\begin{document}

\title[Complex networks and multiple zeta functions]{Bipartite and directed scale-free complex networks arising from zeta functions}

\author{Piergiulio Tempesta}
\address{Departamento de F\'{\i}sica Te\'{o}rica II, M\'etodos Matem\'aticos de la f\'isica, Facultad de F\'{\i}sicas, Ciudad Universitaria, Universidad
Complutense de Madrid, 28040 - Madrid, Spain and Instituto de Ciencias Matem\'aticas, C/ Nicol\'as Cabrera, No 13--15, 28049 Madrid, Spain}
\email{p.tempesta@fis.ucm.es}

\begin{abstract}
We construct a new class of directed and bipartite random graphs whose topology is governed by the analytic properties of multiple zeta functions. The Euler-Zagier and the multiplicative zeta graphs are relevant examples of the proposed construction. Phase transitions and percolation thresholds for our models are determined. 
\end{abstract}
\date{November 15, 2021}
\maketitle

\tableofcontents

\section{Introduction}


In the last decade, complex networks have been acquiring a prominent role in different branches of science as theoretical physics, biology, information science, social sciences, etc. (see, e.g., \cite{Newmanbook}, and the reviews \cite{ABrev}, \cite{Boc}, \cite{Caldbook}, \cite{DG}, \cite{DM}, \cite{Newrev}). Indeed, they are essential in modeling systems with nontrivial interactions, and are usually represented in terms of \textit{random graphs} \cite{Bollobas}, \cite{JLR}. Phenomena like phase transitions in complex networks depend crucially on the topology of the underlying graphs.

Since the pioneering work of Erd\"os and R\'enyi \cite{ER1}, and Solomonoff and Rapoport \cite{SR}, this field has known a dramatic development, and has been widely investigated \cite{Bollobas}. Many new models of random graphs have been considered and their role in the applications analyzed.

In particular, \textit{scale--free} models, i.e. models exhibiting a power-law degree distribution, represent one of the most studied classes of complex networks. Historically, the first example of them was offered by the Price model \cite{Price}. Among the most important ones are those proposed by Barabasi, Albert and collaborators in \cite{BA}, \cite{BAJB} and by Aiello et al. in \cite{Aiello} (for a recent review, see the monograph \cite{Caldbook}).

The aim of this work is to establish a connection between the theory of complex networks and number theory. This research can be considered part of a general program aiming at investigating the relation between statistical mechanics and  number theory. A first connection between these two fields was discovered by Montgomery and Odlyzko: the Gaussian unitary ensemble was related with the zeros of the Riemann zeta function $\zeta(s)$. Since then, many studies have been devoted to clarify the relation among generalized zeta functions, random matrix theory, and various aspects of quantum field theory and spectral theory (see \cite{frontiers} and \cite{NTP} for general reviews).

In \cite{Ltemp}, L-functions, which are among the most important objects in analytical number theory, have been used to construct scale-free graphs possessing several interesting topological properties.   Essentially, L-functions are meromorphic continuations of Dirichlet series to the whole complex plane, with a Euler product and a functional equation (for a modern introduction, see e.g. the monograph \cite{IK}). The classical Riemann zeta function is the most known example of a L-function. An axiomatic theory of L-functions has been proposed by A. Selberg \cite{Selberg}. In \cite{Tempesta3}, Dirichlet series have been related to generalized entropies via the notion of universal formal group (see also \cite{Tempesta1}, \cite{Tempesta2}.)

In this paper, we further extend the results of \cite{Ltemp} and construct new families of \textit{directed and bipartite random graphs with scale similarity properties}. The main motivation is that, apart their intrinsic theoretical interest, graphs of this class appear in many real-world networks.

A directed graph is a graph in which a direction is specified in every link. This makes directed graphs more sophisticated and realistic than undirected graphs. A basic example of directed graph is the world-wide web.

In turn, bipartite graphs are characterized by two different types of vertices, each endowed with a degree distribution. These models are also widely investigated for their usefulness in different applications (\cite{ABrev}, \cite{Newmanbook}).

As is common in the literature, the scale--free invariance of a unipartite model essentially means that, given the probability distribution $\{p_k\}_{k\in\mathbb{N}}$, the ratio $p_{\alpha \times k}/ p_{k} $ depends only on $\alpha$ but not on $k$. In the following, we will consider a natural generalization of this notion for classes of directed and bipartite random graphs.

We are here mainly interested in the topology of the new networks we introduce. In particular, we will study the conditions under which phase transitions may occur. These condition will be typically expressed in terms of functional equations in the parameters of our models, that can be solved numerically with arbitrary accuracy.

The paper has the following structure. In Section 2, some basic definitions concerning the theory of generating functions for random graphs are proposed. In Section 3, bipartite graphs arising from suitable Dirichlet series and L-functions are introduced. In Section 4, an analogous construction is proposed for the case of directed graphs. In Section 5, the relevant subcase of scale-free networks is analyzed in detail, in terms of a suitable group theoretical structure allowing the composition of graphs. In Section 6, an application of our models to a biological context is proposed.

\section{Algebraic Preliminaries}

In order to fix the language and the notation, in this Section some basic aspects of the formalism adopted in the paper will be sketched.

Throughout this work, we will stay in the limit of large graph size.

In order to define a random graph, a degree probability distribution $\{p_n\}_{n\in\mathbb{N}}$ of vertices in the graph is introduced, where $p_n$ is the probability that an uniformly randomly chosen vertex has degree $n$.  Once assigned a degree probability distribution, a graph is chosen uniformly at random in the class of all graphs with that given distribution. We shall  follow the analysis and  the terminology of \cite{Newrev} and \cite{Newman1}.

Consider first the case of an \textit{undirected graph}. The series
\beq
G_{0}(x):=\sum_{n=0}^{\infty}p_n x^n, \label{GF}
\eeq
is called the generating function of the distribution. We have necessarily
\beq
G_{0}(1)=1.
\eeq
The distribution of the outgoing edges is generated by
\beq
G_{1}(x):=\frac{\sum_{n=1}np_{n}x^{n-1}}{\sum_{n=1}np_{n}}, \label{g1}
\eeq
where the average number of first neighbors, equal to the average degree of the graph, is
\begin{equation*}
z_1= <n>=\sum_{n} n p_n= G'_{0}(1).
\end{equation*}

In the case of \textit{directed graphs}, each vertex possesses an in-degree $j$ and an out-degree $k$. Therefore, one introduces a distribution $\{\pi_{jk}\}_{j,k\in\mathbb{N}}$ over both degrees. The generating function for a directed graph is of the form
\beq
G(x,y)=\sum_{j,k}\pi_{jk}x^jy^k. \label{gf2}
\eeq

It is natural to introduce generating functions for the in-degrees and out-degrees, which are obtained from \eqref{gf2} by summing away the irrelevant degrees of freedom:
\beq
F_{0}(x)=G(x,1); \hspace{5mm} G_{0}(y)=G(1,y).
\eeq

In a \textit{bipartite graph}, we can distinguish two kinds of vertices, with edges running only between vertices of unlike types. We will call them of type $A$ and of type $B$, respectively (we shall consider the case of a very large number of vertices).. We introduce two degree distributions,  $\{p_n\}_{n\in\mathbb{N}}$ and  $\{q_m\}_{m\in\mathbb{N}}$, corresponding to the probability distributions of the nodes of type $A$ and $B$, respectively.  The associated generating functions will be
\beq
f_{0}(x)=\sum_{n}p_{n}x^n, \hspace{5mm} g_{0}(x)=\sum_{m}q_{m}x^{m}.
\eeq

Let us suppose that we have $a$ vertices of type $A$ and $b$ vertices of type $B$, and that each vertex of type $A$ has an average of $\overline{a}$ links with nodes of type $B$ and that each vertex of type $B$ has $\overline{b}$ links with nodes of type $A$. These quantities are clearly related by the constraint
\beq
a \overline{b}= b \overline{a}.
\eeq

\noindent According to the previous discussion, we have that
\beq
f_{0}(1)=g_{0}(1)=1
\eeq
and
\beq
f'_{0}(1)=\overline{a}, \qquad g'_{0}(1)=\overline{b}.
\eeq

\noindent A condition analogous to \eqref{g1} is
\beq
f_{1}(x)=\frac{f_{0}'(x)}{\overline{a}}, \hspace{5mm} g_{1}(x)=\frac{g_{0}'(x)}{\overline{b}}.
\eeq

In the coming sections, more specific properties, as topological phase transitions, will be discussed for both directed and bipartite graphs.

\section{Bipartite graphs coming from L--functions}

\subsection{Main definition}

We introduce a very large class of models, the bipartite $L$-graphs. In the subsequent discussion, we shall consider some Dirichlet series with specific properties, that will make them suitable for the construction proposed.

\begin{definition} A \textit{bipartite L-graph} is a complex network characterized by the following two degree probability distributions for the nodes of type $A$ and $B$ respectively:
\begin{equation}
\nn p_{m}= \begin{cases} 0 \qquad\qquad\qquad for \quad m=0, \\
\frac{a_m m^{-\alpha}}{L_1(\alpha)}\qquad \quad for \quad m\in\mathbb{N}, \quad m\geq 1, \label{L1}
\end{cases}
\end{equation}
\begin{equation}
q_{n}= \begin{cases} 0 \qquad\qquad\qquad for \quad n=0, \\
\frac{b_n n^{-\beta}}{L_2(\beta)} \qquad\qquad\quad for\quad n\in\mathbb{N}, \quad n\geq 1. \label{L2}
\end{cases}
\end{equation}
Here
\beq
\nn L_1(\alpha)=\sum_{m=1}^{\infty}\frac{a_m}{m^{\alpha}}, \qquad L_2(\beta)=\sum_{n=1}^{\infty}\frac{b_n}{n^{\alpha}}, \qquad \alpha,\beta\in\mathbb{R}, \quad \alpha,\beta>1
\eeq
\noindent are suitable Dirichlet series ($L$--series), and $p_m$, $q_m\geq 0$. 
\end{definition}

\begin{remark}
In the following, we shall tacitly assume that $L_1(\alpha)$, $L_2(\beta)$ take real values; also, the coefficients $a_m$, $b_n$ are supposed to be nonnegative, with $a_m\neq 0$, $b_n\neq 0$ for some $m,n$. However, in Section 5.2 we will also consider the formal situation corresponding to the choice $a_m$, $b_n\in \mathbb{R}$ for all $m,n$.
\end{remark}

One of the most important topological properties of a random graph is the possible formation of a giant cluster. It corresponds to a topological phase transition, marked by a threshold value in one of its parameters.  Given a unipartite random graph with $N$ vertices, in \cite{MR} a threshold condition has been determined almost surely, i.e. with probability tending to 1 for $N\rightarrow\infty$, for degree probability distributions well behaved and \textit{sufficiently regular}.
This condition has been slightly weakened in  \cite{HM}.

In many cases, it is useful to consider the one-mode network \cite{WF}, which is the projection of the bipartite graph onto the unipartite space of the vertices of type $A$ (or $B$) only.

In order to find the analytic condition for the transition, let us introduce the function
\beq
G_{1}(x)=f_{1}(g_{1}(x)).
\eeq

\noindent It plays the same role as the function \eqref{g1} for unipartite graphs. One can prove \cite{Newman1} that the giant component first appears when
\beq
f_{0}''(1)g_{0}''(1)=f_{0}'(1)g_{0}'(1). \label{cond}
\eeq
By using eq. \eqref{cond}, we can state the following result.
\begin{theorem} \label{P1}
Given a bipartite $L$-graph, the threshold condition marking the phase transition to the formation of a giant cluster occurs for the values of $(\alpha, \beta)$ such that
\bea
\noindent\nn \Psi(\alpha,\beta)&=&L_{1}(\alpha-2)L_{2}(\beta-2)-L_{1}(\alpha-2)L_{2}(\beta-1)\\
\noindent &&-L_{1}(\alpha-1)L_{2}(\beta-2)=0. \label{threscond}
\eea
\end{theorem}
\begin{proof}
Condition \eqref{cond} translates into the threshold equation
\beq
\sum_{m, n=1}^{\infty} mn(mn-m-n)p_{m}q_{n}=0. \label{phasetr}
\eeq
The relation \eqref{phasetr}, for bipartite $L$--graphs is equivalent to eq. \eqref{threscond}.
\end{proof}

\noindent The threshold condition \eqref{threscond} is expressed in terms of a functional equation in two parameters, whose solution space lies in a two-dimensional surface, defining a phase diagram. In agreement with the results of \cite{MR}, it is necessary for the giant cluster to exist in the bipartite case that
\beq
\Psi(\alpha,\beta)>0. \label{ineq}
\eeq
For unipartite undirected graphs, this requirement, joint with the growth condition for $k_{max}$, implies the existence of a giant cluster \cite{MR}.
\begin{remark}
One can prove that a good estimation for the asymptotic growth of the maximum degree for a power law-type distribution is $k_{max}\sim N^{\frac{1}{\alpha-1}}$ \cite{Newrev}. This enables to obtain sufficient regularity in asymptotic growth for the unipartite projections of the $L$-models by further constraining the parameter space.
\end{remark}

\noindent Observe that when $L_1(\alpha)=L_2(\beta)=L(\alpha)$, for the unipartite zeta graphs the previous proposition reduces to
\beq
L(\alpha-2)-2L(\alpha-1)=0. \label{unipthr}
\eeq
The threshold equation \eqref{unipthr} has been first obtained in \cite{Ltemp}. It reduces to the condition for the first appearance of the giant component of the model by Aiello et al., i.e. $\zeta(\alpha-2)=2 \zeta(\alpha-1)$ \cite{Aiello}.

Another important topological property of a graph is the \textit{clustering coefficient}, or network transitivity, defined as
\begin{equation}
C:=\frac{\text{3x number of triangles in the network}}{\text{number of connected triples of vertices}},\label{clustering}
\end{equation}
where a connected triple consists of a single vertex whose edges connect it to an unordered pair of others. The clustering coefficient satisfies $0\leq C\leq 1$.


One can prove by a direct calculation the following general statement.
\begin{proposition}
For a bipartite $\zeta$-graph, the clustering coefficient \eqref{clustering} is expressed by the formula
\beq
C=\frac{L_{1}(\alpha-1)L_{2}(\beta-1) F(\beta)}{[L_{1}(\alpha-2)-L_{1}(\alpha-1)][L_{2}(\beta-2)-L_{2}(\beta-1)]^2 +1},
\eeq
where
\beq
F(\beta)=[2L_2(\beta-1)-3 L_2(\beta-2)+L_2(\beta-3)].
\eeq
\end{proposition}
\subsection{The bipartite zeta graph}
A first example of the class of bipartite $L$-graphs is the $\zeta$-graph, directly related to the Riemann zeta function. It generalizes the celebrated model by Aiello et al. \cite{Aiello}.

\begin{definition} A \textit{bipartite $\zeta$-graph} is a complex bipartite L-graph generated by two distribution functions of the form
\beq
p_{n}=\frac{ n^{-\alpha}}{\zeta(\alpha)} \quad n\in\mathbb{N}/ \{0\}, \quad \alpha \geq 1,  \qquad q_{m}=\frac{ m^{-\beta}}{\zeta(\beta)}\quad m\in\mathbb{N}/ \{0\}, \quad \beta \geq 1,
\eeq
with $p_0=q_0=0$, where
\beq
\zeta(\alpha)=\sum_{k=1}^{\infty}\frac{1}{k^{\alpha}}, \qquad \alpha>1
\eeq
is the Riemann zeta function.
\end{definition}
\noindent For the $\zeta$-model we get easily
\beq
\overline{a}=\overline{b}=\frac{\zeta(\alpha-1)}{\zeta(\alpha)}.
\eeq
We obtain for the phase transition the threshold condition
\beq
\zeta(\alpha-2)\zeta(\beta-2)=\zeta(\alpha-2)\zeta(\beta-1)+\zeta(\alpha-1)\zeta(\beta-2).
\eeq

\subsection{The Hurwitz graph}
The Hurwitz bipartite graph is a random graph model related to the classical Hurwitz zeta function
\begin{equation*}
\zeta_{H}(s,k_0)=\sum_{k=1}^{\infty}\frac{1}{(k+k_0)^{s}}, \quad s\in\mathbb{C}, \quad \text{Re} \hspace{1mm} s >1.
\end{equation*}

\begin{definition}
A bipartite random graph whose degree probability distributions are of the form
\begin{equation}
\noindent \nonumber
\begin{cases} p_0=0,\qquad q_0=0, \\
p_m=\frac{(k_0+m)^{-\alpha}}{\zeta_{H}(\alpha)}, \quad m\in\mathbb{N}, \quad \alpha\in\mathbb{R}, \quad \alpha>1, \quad q_n=\frac{(k_0+n)^{-\beta}}{\zeta_{H}(\beta)}, \quad n\in\mathbb{N}, \quad \beta\in\mathbb{R}, \quad \beta>1,
\end{cases}
\end{equation}
will be called a \textit{Hurwitz bipartite random graph}.
\end{definition}

We remind \cite{Ltemp} that the projected one-mode of this graph can be related with \textit{nonextensive statistical mechanics} \cite{Tsallisbook}. This formulation of statistical mechanics is based on a generalization of the Boltzmann-Gibbs entropy, i.e. the Tsallis entropy
\begin{equation*}
S_q=  \frac{1-\sum_{i=1}^{W}p_{i}^{q}}{1-q}, \quad i=1, \ldots, W.
\end{equation*}
Let $e_{q}(x):=[1+(1-q)x]^{\frac{1}{1-q}}$ denote the $q$-exponential function. Formally, by writing the distribution $p_m=\frac{(k_0+m)^{-\alpha}}{\zeta_{H}(\alpha)}$ in terms of $e_{q}(-m/\tau)$, and putting $\alpha=\frac{1}{q-1}$, $k_0=\frac{\tau}{q-1}$, we essentially get the optimizing distribution for the Tsallis entropy, arising in the description of the stationary state associated with the canonical ensemble in the nonextensive scenario \cite{Tsallisbook}, \cite{TsallisEPJ}.

By applying Theorem \ref{P1}, we deduce that the phase transition is described by the critical equation
\beq
\zeta_{H}(\alpha-2, k_0)\zeta_{H}(\beta-2, k_0)=\zeta_{H}(\alpha-2, k_0)\zeta_{H}(\beta-1, k_0)+\zeta_{H}(\alpha-1, k_0)\zeta_{H}(\beta-2, k_0),
\eeq
which depends on the entropic index $q$ (for simplicity, we put $\tau=1$). It generalizes the condition $\zeta_{H}(\alpha-2, k_0)=2 \zeta_{H}(\alpha-1, k_0)$, valid for the unipartite model.

\section{Directed graphs and L-functions}

We wish to introduce here directed graphs related to \textit{L- functions}. 


To specify the degree of each vertex, we introduce a couple $(n,m)$ of natural integers, expressing the in--degree and the out--degree, respectively.
\noindent The associated degree probability distribution $\{\pi_{n,m}\}_{n,m\in\mathbb{N}}$ must satisfy the consistency constraint
\beq
\sum_{n,m}(n-m)\pi_{nm}=0, \label{zerovert}
\eeq
expressing the fact that the net average number of edges entering a vertex is zero. The constraint \eqref{zerovert} implies that the average in--degrees and out--degrees of the vertices coincide, and are given by
\beq
z=\sum_{n}n \pi_{nm}.
\eeq
The threshold condition for directed graphs reads
\beq
\sum_{n,m}(2 n m - n - m)\pi_{nm}=0  \label{thr}
\eeq
There are several choices available in order to define directed graphs from L-functions. In the following, we discuss some possible constructions.

\subsection{A class of separated models}

An interesting class of directed random graphs is the separated one, obtained from degree distributions of the form
\beq
\pi_{nm}=p_{n}q_{m}. \label{separ}
\eeq
This hypothesis is certainly restrictive. However, an important example of directed graph fulfilling eq. \eqref{separ} is the world-wide web. It has been shown that a good approximation of the experimental data can be obtained by means of the choice
\beq
p_n=\frac{(n+k_0)^{-\alpha}}{\zeta_{H}(k_0,\alpha)},
\eeq
with $q_m$ of the same functional form.
\noindent More generally, by taking two copies of a unipartite L-graph, and using formula \eqref{separ},  we can produce a directed graph, that we shall call a \textit{directed separated L-graph}. The constraint \eqref{zerovert} is automatically satisfied. Instead, the threshold condition \eqref{thr} imposes the further constraint
\beq
2 L(\alpha-1)L(\beta-1)-L(\alpha-1)L(\beta)-L(\alpha)L(\beta-1)\geq 0. \label{direct}
\eeq
This condition generalizes that one valid for the world wide web, obtained by identifying $L(\alpha)$ with $\zeta_{H}(\alpha)$.

We propose a definition of scale invariance for this class of graphs.
\begin{definition}
We shall say that a directed separated graph is scale-invariant if each of the distributions $\{p_n\}_{n\in\mathbb{N}}$ and $\{q_m\}_{m\in\mathbb{N}}$ is scale-invariant.
\end{definition}

In the following section, an algebraic approach for the generation of scale-invariant directed L-graphs will be proposed.

\textbf{Remark}. The theory of multiple zeta functions dates to the works of Euler, and of Barnes and Mellin at the beginning of the 20th century. A renewal of interest in the field started with the works of Zagier \cite{Zagier} and Hoffman \cite{Hoffman}. 
A multiple zeta function very common in the literature is the \textit{Barnes zeta function}, defined to be
\beq
\zeta_{B}(s,w \mid a_1,a_2):=\sum_{m,n=0}^{\infty}(w+m a_1 + n a_2 )^{-s}, \text{Re s}>2, \text{Re} \hspace{1mm} a_{1}, \text{Re} \hspace{1mm} a_{2}, \text{Re w}>0.
\eeq
By analogy with the previous definitions (under the assumption of Remark 3), one can introduce a directed graph whose probability distribution is given by
\begin{equation}
\pi_{nm}= \begin{cases} 0 \qquad\qquad\qquad\quad for \quad n=m=0, \\
[w+(n+m)a]^{-\alpha}/\zeta_{B}(\alpha,w \mid a,a)\qquad for\quad n,m \geq 1,\text{and} \hspace{1mm} a, \alpha>0. \label{zeta}
\end{cases}
\end{equation}
This definition is well posed, since the consistency condition \eqref{zerovert} is satisfied. However, apparently there is no easy way to write the phase transition threshold condition in a close form as a functional equation, which makes the model less transparent and of a difficult treatment. It would be interesting to construct graphs related to the theory of multiple zeta functions of Zagier-Hoffman type.

\section{Multiplicative zeta functions and related scale-free networks}
\subsection{Algebraic preliminaries}

In \cite{Ltemp}, a family of scale-free random graphs has been constructed by means of the analytic theory of multiplicative functions \cite{Apostol}.
First, we propose a definition of scale invariance in this context.
\begin{definition}
We shall say that a bipartite graph is scale-invariant if each of the distributions $\{p_n\}_{n\in\mathbb{N}}$ and $\{q_m\}_{m\in\mathbb{N}}$ is scale-invariant.
\end{definition}
Here we briefly recall some basic facts about the theory of multiplicative functions.

An application $f:\mathbb{N}\rightarrow\mathbb{R}$, not identically zero is said to be a multiplicative arithmetic function if
\begin{equation*}
f(mn)=f(m)f(n) \quad \text{whenever} \quad (m,n)=1.
\end{equation*}
A necessary condition for $f$ to be multiplicative is that
\[
f(1)=1.
\]
In particular, the function $f$ will be said to be completely multiplicative if
\begin{equation*}
f(mn)=f(m)f(n) \quad \text{for all} \quad m,n\in\mathbb{N}.
\end{equation*}
Well--known examples of multiplicative functions are the Euler totient function and the $\chi$--function. A completely multiplicative function is the Liouville one.

If $f$ is completely multiplicative, then there exists a constant $\sigma_a\in\mathbb{R}$ such that the series
\begin{equation}
\sum_{n=1}^{\infty}\frac{f(n)}{n^{s}}, \quad s\in\mathbb{C} \label{f--funct}
\end{equation}
converges absolutely for $\text{Re s}>\sigma_{a}$. Also, we have the Euler product representation
\begin{eqnarray}
\sum_{n=1}^{\infty}\frac{f(n)}{n^{s}}=\prod_{p}\left\{1+\frac{f(p)}{p^{s}}+\frac{f(p^2)}{p^{2s}}+\cdots \right\}, \qquad \text{Re s}>\sigma_{a},
\end{eqnarray}
if $f$ is multiplicative, and
\begin{equation}
\sum_{n=1}^{\infty}\frac{f(n)}{n^{s}}=\prod_{p}\frac{1}{1-f(p)p^{-s}}, \quad \quad \text{Re s}>\sigma_{a},
\end{equation}
if $f$ is completely multiplicative.

\noindent We can now consider a specific class of bipartite random graphs.

\begin{definition} \label{cmbip}
A bipartite L-graph is said to be (completely) multiplicative if the $L$-series \eqref{L1} are of the form \eqref{f--funct}, where $f$ is a (completely) multiplicative function.
\end{definition}
\noindent The previous definition easily extends to the case of a multipartite $L$-graph, i.e., a graph in which several distinct sets of nodes are present, each of them represented by a probability distribution of the form $p_n=a_ n n^{-\alpha}/L(\alpha), n\in\mathbb{N}/\{0\}$.

\noindent A consequence of the Definition \ref{cmbip} is the following result.
\begin{corollary}\label{corollary1}
A completely multiplicative bipartite $L$-graph is scale-free.
\end{corollary}

\subsection{A product in the space of bipartite and directed graphs}

As in  \cite{Ltemp}, we introduce  the space $\mathcal{G_{\mathcal{M}}}$ of multiplicative unipartite zeta random graphs, and  the space $\mathcal{G_{\mathcal{CM}}}$ of completely multiplicative unipartite random graphs. We will show that these spaces play a special role in the construction of new bipartite and directed complex networks. To this aim, we define first a group-theoretical structure.

\begin{definition}
The product of two unipartite multiplicative random graphs $G_1(f)$ and $G_2(g)$ is defined to be the graph $G_{12}(h)$ whose associated multiplicative function is the Dirichlet convolution of $f$ and $g$:
\begin{equation}
h:=(f*g)(mn)=\sum_{x \mid mn} f(x) g\left(\frac{mn}{x}\right). \label{Dirprod}
\end{equation}
\end{definition}
As is well known \cite{Apostol}, the convolution function $h$ is also \textit{multiplicative}. In addition,  we can define the inverse of a graph with respect to the product \eqref{Dirprod}.
In this construction, the M\"obius function $\mu(n)$ plays a prominent role. It is defined as follows. Given $n\in\mathbb{N}$, first we write it in the form $n=p_{1}^{a_1}\cdots p_{k}^{a_k}$, with $p_1,\ldots,p_k$ suitable prime numbers.

\noindent Then we put

\begin{equation*}
\mu(n)= 1 \qquad\qquad\qquad \text{if} \quad n=1;
\end{equation*}
\begin{equation*}
\text{for} \quad n>1, \quad \mu(n)= \begin{cases} (-1)^{k} \quad \text{if} \quad a_1=a_2=\cdots=a_{k}=1\\ \label{mu}
\quad 0 \qquad \text{otherwise}.
\end{cases}
\end{equation*}
In other words, $\mu(n)=0$ if and only if $n$ has a square factor $>1$. The M\"obius function is related to the Euler function $\phi(n)$ by the formula $\phi(n)=\sum_{d|n}\mu(d)\frac{n}{d}$, $n\geq1$.

\begin{definition}
The ``inverse'' graph of a completely multiplicative unipartite random graph $G(f)$ is formally defined to be the graph $G^{-1}(f):=G(f^{-1})$ whose associated multiplicative function is
\begin{equation*}
f^{-1}(n):=\mu(n)f(n) \quad \forall \hspace{1mm} n\geq 1.
\end{equation*}
\end{definition}

\begin{remark} A priori, in this definition $\mu(n)$ should be nonnegative. However, in principle we can consider a more general class of \textit{pseudographs}, where $f: \mathbb{R}\to \mathbb{Z}$. More precisely, in this case the degree distribution is given by a set of quasiprobabilities (a quasidistribution). Thus, we can define a formal inverse of a graph with respect to the product. 
\end{remark}

The space $\mathcal{G_{\mathcal{M}}}(*)$ has the structure of a  monoid. If we allow the set $\mathcal{G_{\mathcal{M}}}$ to include pseudographs, we get the structure of a group. 

\begin{remark}Given a degree quasidistribution, we can associate with it a standard degree distribution (and, consequently, a set of standard graphs).  A direct way consists in taking the absolute value of the quasiprobabilities and nomalizing them again properly. Another possibility is to exclude the negative degrees and to further normalize the remaining nonnegative quantities. Clearly, these choices give rise to different families of graphs.
\end{remark}

Observe that if we restrict to $\mathcal{G_{\mathcal{CM}}}$, the Dirichlet product is not necessarily completely multiplicative. Another possibility \cite{Ltemp} is to consider the pointwise product of two completely multiplicative functions: $h:=(f\cdot g)(n)=f(n)g(n)$ and define the product of graphs as the graph associated with $h$. The space $\mathcal{G_{\mathcal{CM}}}(\cdot)$ is now an abelian monoid: \textit{scale-free networks are transformed into scale-free networks}.



By way of an example, we will construct the inverse of the bipartite version of the model of Aiello et al. \cite{Aiello}: the bipartite M\"obius random graph. The Dirichlet series associated with $\mu(n)$ is
\begin{equation*}
\sum_{n=1}^{\infty}\frac{\mu(n)}{n^{s}}=\frac{1}{\zeta(s)}=\prod_{p}(1-p^{-s}), \quad \text{if} \quad \text{Re s}>1.
\end{equation*}
Therefore, we propose the following definition.
\begin{definition} A \textit{bipartite M\"obius graph} is a complex network characterized by the following two degree distributions for the nodes of type $A$ and $B$ respectively:
\begin{equation}
\noindent \nonumber \begin{cases} p_0=q_0=0 \\
p_k=q_k=\frac{\mu(k)\zeta(k)}{k^{\alpha}}, \quad k\in\mathbb{N}/\{0\}, \quad \alpha>1. \end{cases}
\end{equation}
\end{definition}
The algebraic formalism proposed above allows us to define the product of multipartite graphs. The bipartite version is the following.
\begin{definition}
The Dirichlet (or pointwise) product of two bipartite multiplicative L-graphs is the bipartite L-graph whose degree distributions for the nodes of type A and B are defined by the Dirichlet (or pointwise) product of the corresponding distributions of the two graphs.
\end{definition}


The same idea can be used to ``multiply" directed graphs of a suitable type. Consider the case of a probability distribution of the form $\pi_{nm}=p_{n}q_{m}$, as for the case of the world-wide web. Instead of considering a  Hurwitz zeta distribution, we can assume that $p_n$ and $q_m$ are represented by a (completely) multiplicative zeta function. In this way we define the class of directed separated multiplicative $L$-graphs. For this set, a result similar to Corollary \ref{corollary1} holds.
\begin{corollary}\label{corollary2}
A completely multiplicative directed separated $L$-graph is scale-free.
\end{corollary}
Also, we can introduce a similar notion of product of directed $L$-graphs.
\begin{definition}
The Dirichlet (or pointwise) product of two directed  multiplicative separated $L$-graphs is the directed graph whose degree distributions for the nodes of type A and B are defined by the Dirichlet (or pointwise) product of the corresponding distributions of the two graphs.
\end{definition}

\section{Epidemic transitions for structured populations}

The networks previously proposed can be  used as population models for the study of phenomena like the spread of epidemics. Here we focus on SIR (susceptible/infective/recovered) models \cite{AM}, \cite{Murray}. Starting from the work \cite{Grass}, the connection between percolation, epidemiology and complex networks has been addressed by several authors \cite{Call}, \cite{CEbH}, \cite{EK}, \cite{PSV1}, \cite{PSV2}. In these models, the population is divided into three possible states, S, I and R. The state R can also indicate the ``removal" of an individual (due to recovery or death).
We can think of the bipartite zeta graphs as examples of bipartite populations, where the degree distributions $f_0(x)$ and $g_0(x)$ of two groups of individuals, say $m$ and $f$ (for instance, males and females), are assigned. Following the notation of \cite{Newspread}, we introduce the transmissibility coefficients $T_{mf}$ and $T_{fm}$ of a given disease in the two directions. We also introduce the generating functions
\beq
F_{0}(x;T_{mf},T_{fm})=f_{0}(g_{1}(x;T_{fm});T_{mf})
\eeq
and
\beq
F_{1}(x;T_{mf},T_{fm})=f_{1}(g_{1}(x;T_{fm});T_{mf}).
\eeq
The average outbreak size for individuals of the group $m$ is
\beq
\langle s \rangle =1+ \frac{F_{0}'(1;T_{mf},T_{fm})}{1-F_{1}'(1;T_{mf},T_{fm})} \label{smale},
\eeq
which gives the epidemic threshold condition
\beq
T_{mf}T_{fm}=\frac{1}{f_{1}'(1)g_{1}'(1)}. \label{etc}
\eeq
Notice that this result is symmetric in the variables corresponding to the properties of the $m$ and $f$ populations: only the product of the transmissibilities is relevant.

We have the following general result.
\begin{theorem}
The epidemic threshold for a bipartite $L$-graph \eqref{L1} is given by
\beq
T_{mf}T_{fm}=\frac{L_{1}(\alpha-1)L_2(\beta-1)}{[L_1(\alpha-2)-L_1(\alpha-1)][L_2(\beta-2)-L_2(\beta-1)]}. \label{epid}
\eeq
\end{theorem}
\begin{proof}
For a bipartite zeta model, we have that
\beq
\bar{a}=\frac{L_1(\alpha-1)}{L_1(\alpha)}, \qquad \bar{b}=\frac{L_2(\beta-1)}{L_2(\beta)}.
\eeq
Also
\beq
f'_{1}(1)=\frac{\sum_{k}k(k-1)p_k}{\bar{a}}= \frac{L_{1}(\alpha-2)-L_{1}(\alpha-1)}{L_{1}(\alpha-1)},
\eeq
and similarly for $g_{1}'(1)$. Condition \eqref{etc} gives us eq. \eqref{epid}.
\end{proof}
An interesting particular case is obtained when the power laws for the two populations are the same, with $a_k=b_k=1$ for all $k$. In this case we get the known formula $T_{mf}T_{fm}=T_{c}^2$, with
\beq
T_{c}= \frac{\zeta(\alpha-1)}{\zeta(\alpha-2)-\zeta(\alpha-1)},
\eeq
discussed in \cite{Newspread}.

Many other aspects of the models proposed in this work can be studied numerically. Also, several other research lines deserve to be investigated, as, for instance, the analytic properties of the adjacency matrices or the Laplacian operators associated with these models. It would be also important to obtain growing models possessing bipartite or directed zeta graphs as limiting configurations.

\textbf{Acknowledgments}

The research of P. T. has been partly supported by the grant FIS2011--00260, Ministerio de Ciencia e Innovaci\'on, Spain.

\bigskip

\end{document}